\documentclass[conference]{IEEEtran}
\usepackage{mathrsfs,amsmath,amssymb,epsfig,amsfonts,fancyhdr,graphicx,cite,amsthm}
\usepackage{engord}

\newtheorem{thm}{Theorem}[]
\newtheorem{cor}{Corollary}[]

\theoremstyle{remark}
\newtheorem{rem}[]{Remark}

\theoremstyle{definition}

\ifCLASSINFOpdf
\else
\fi
\begin{document}
%
\title{Combined Decode-Forward and Layered Noisy Network Coding Schemes for Relay Channels}

\author{\authorblockN{Peng Zhong and Mai Vu\\}
\authorblockA{Department of Electrical and Computer Engineering, McGill University, Montreal\\
Emails: peng.zhong@mail.mcgill.ca, mai.h.vu@mcgill.ca}}


%



\maketitle

\begin{abstract}

We propose two coding schemes  combining decode-forward (DF) and noisy network coding (NNC) with different flavors. The first is a combined DF-NNC scheme for the one-way relay channel which includes both DF and NNC as special cases by performing rate splitting, partial block Markov encoding and NNC.
The second combines two different DF strategies and layered NNC   for the two-way relay channel. One DF strategy performs coherent block Markov encoding at the source at the cost of power splitting at the relay, the other performs independent source and relay encoding but with full relay power, and layered NNC allows a different compression rate for each destination.    Analysis and simulation show that both proposed schemes   supersede each individual   scheme and take full advantage of both DF and NNC.

\end{abstract}


%
\IEEEpeerreviewmaketitle

\section{Introduction}
The relay channel (RC) first introduced by van der Meulen  consists of a source aiming to communicate with a destination with the help of a relay. In \cite{cover1979capacity}, Cover and El Gamal propose the fundamental decode-forward (DF), compress-forward (CF) and combined DF-CF schemes. 
Lim, Kim, El Gamal and Chung recently put forward a noisy network coding (NNC) scheme \cite{sung2011noisy} for the general multi-source network. NNC is based on CF relaying but involves three new ideas (message repetition, no Wyner-Ziv binning and simultaneous decoding) and outperforms CF for multi-source networks. In \cite{Ramalingam2011superposition}, Ramalingam and Wang propose a superposition NNC scheme for  restricted relay networks, in which source nodes cannot act as relays, by combining  DF and NNC and show some performance improvement over NNC. Their scheme, however, does not include DF relaying rate because of no block Markov encoding.\par

The classical one-way relay channel can be generalized to the two-way relay channel (TWRC), in which two users exchange messages with the help of a relay. In \cite{rankov2006achievable}, Rankov and Wittneben apply several relay strategies, including decode-forward and compress-forward,  to the TRWC. In their proposed DF scheme, the two users perform partial block Markov encoding, and the relay sends a superposition of the codewords for the two decoded messages in each block. A different DF strategy is proposed in \cite{xie2007network} by Xie, in which the users encode independently with the relay without block Markovity, and the relay sends a codeword for the random binning of the two decoded messages. These two DF schemes do not include each other in general. In \cite{Lim2010layered}, Lim, Kim, El Gamal and Chung propose an improved NNC scheme termed  "layered noisy network coding" (LNNC). The relay compresses its observation into two layers: one is used at both destinations, while the other is only used at one destination.\par

In this paper, we first propose a combined DF-NNC scheme for the one-way  channel.
Different from \cite{Ramalingam2011superposition}, our proposed scheme performs block Markov encoding and hence encompasses both DF relaying and NNC as special cases. We then propose a combined DF-LNNC scheme for the TWRC. This scheme also includes partial block Markov encoding and, in addition, performs layered NNC. Analysis and numerical results show that this  scheme outperforms each individual scheme in \cite{rankov2006achievable, xie2007network, Lim2010layered} and also the combined scheme in \cite{Ramalingam2011superposition}.

\section{Channel Models}

\subsection{Discrete memoryless relay channels}
The discrete memoryless two-way relay channel (DM-TWRC) is
denoted by $(\mathcal{X}_1 \times \mathcal{X}_2 \times
\mathcal{X}_r,$ $p(y_1,y_2,y_r|x_1,x_2,x_r),$ $\mathcal{Y}_1 \times
\mathcal{Y}_2 \times \mathcal{Y}_r)$ where $(x_1, y_1), (x_2,y_2), (x_r, y_r)$ are input and output signals of user 1, user 2 and the relay, respectively.
A $(n,2^{nR_1},2^{nR_2},P_e)$ code for a DM-TWRC consists of two
message sets $\mathcal{M}_1=[1:2^{nR_1}]$ and
$\mathcal{M}_2=[1:2^{nR_2}]$, three encoding functions
$f_{1,i},f_{2,i},f_{r,i}$, $i=1, \ldots, n$ and two decoding function
$g_1,g_2$ defined as follows:
\begin{align}
&x_{1,i}=f_{1,i}(M_1, Y_1^{i-1}),x_{2,i}=f_{2,i}(M_2,Y_2^{i-1}),x_{r,i}=f_{r,i}(Y_r^{i-1}),\nonumber\\
&g_1:\mathcal{Y}^n_1 \times \mathcal{M}_1 \rightarrow \mathcal{M}_2,
g_2:\mathcal{Y}^n_2 \times \mathcal{M}_2 \rightarrow \mathcal{M}_1\nonumber
\end{align}
The definitions for error probability, achievable rates and capacity follow standard ones in \cite{gamal2010lecture}.\par
The one-way relay channel can be seen as a special case of the TWRC by setting $M_2 = X_2 = Y_1 = \emptyset, f_2 = g_1 = \text{null}$ and $R_2 = 0$.

\subsection{Gaussian one-way and two-way relay channels}
The Gaussian one-way relay channel can be modeled as
\begin{align}
\label{Gaussian_one_model}
  Y=gX+g_{2}X_r+Z, \quad
  Y_r=g_{1}X+Z_r,
\end{align}
where $Z, Z_r\sim\mathcal{N}(0,1)$ are independent Gaussian noises, and $g, g_{1}, g_{2}$ are the corresponding channel gains.\par
The Gaussian two-way relay channel can be modeled as
\begin{align}
  Y_1&=g_{12}X_2+g_{1r}X_r+Z_1 \nonumber \\
  Y_2&=g_{21}X_1+g_{2r}X_r+Z_2 \nonumber \\
  Y_r&=g_{r1}X_1+g_{r2}X_2+Z_r,
  \label{GTWRC}
\end{align}
where $Z_1, Z_2, Z_r\sim\mathcal{N}(0,1)$ are independent Gaussian noises and $g_{12}, g_{1r}, g_{21}, g_{2r}, g_{r1}, g_{r2}$ are the corresponding channel gains. For both channels, the average
input power constraints at each user and the relay are all
$P$.

\section{One-way Relay Channel}
In this section, we propose a coding scheme combing decode-forward\cite{cover1979capacity} and noisy network coding\cite{sung2011noisy} for the one-way relay channel. The source splits its message into two parts, a common   and a private message. The common message is different in each block and  is decoded at both the relay and destination as in decode-forward, while the private message is the same for all blocks and is decoded only at the destination as in noisy network coding.  The source encodes the common message with block Markovity, then superimposes the private message on top.  The relay decodes the common message at the end of each block and compresses the rest as in NNC.  In the next block, it sends a codeword which encodes both the compression index and the decoded common message of the previous block. The destination decodes each common message by forward sliding-window decoding over two consecutive blocks. Then at the end of all blocks, it decodes the private message by simultaneous decoding over all blocks. Our proposed scheme includes both DF relaying and NNC as special cases and outperforms  superposition NNC in \cite{Ramalingam2011superposition} in that we use block Markov encoding for the common messages, which provides coherency between source and relay and improves the transmission rate.

\subsection{Achievable rate for the DM one-way relay channel}
\begin{thm}
\label{rate_one}
The rate $R=R_{10}+R_{11}$ is achievable for the one-way relay channel by combining decode-forward and noisy network coding
\begin{align}
R_{10}&\leq \min \{I(Y_r;U|U_r,X_r), I(U;Y|U_r,X_r)+I(U_r;Y)\}\nonumber \\
R_{11}&\leq \min \{I(X;Y,\hat{Y}_r|U,U_r,X_r),I(X,X_r;Y|U,U_r)\nonumber\\
\label{one_way_bounds}
&~~~~~~~~~-I(\hat{Y}_r;Y_r|X_r,U,U_r,X,Y)\}
\end{align}
for some joint distribution that factors as
\begin{align}
\label{one_way_distribution}
&p(u_r)p(u|u_r)p(x|u,u_r)p(x_r|u_r)\nonumber\\
&p(y,y_r|x,x_r)p(\hat{y}_r|y_r,u,u_r,x_r).
\end{align}
\end{thm}

\begin{proof}
We use a block coding scheme in which each user sends $b-1$ messages over $b$ blocks of $n$ symbols each.
\subsubsection{Codebook generation}
Fix a joint distribution as in \eqref{one_way_distribution}. For each block $j\in[1:b]$:
\begin{itemize}
\item Independently generate $2^{nR_{10}}$ sequences $u_{r,j}^n(m_{j-1})\sim\prod ^n_{i=1}p(u_{r,i}) $, where $m_{j-1} \in [1:2^{nR_{10}}]$.
\item For each $m_{j-1}$, independently generate $2^{nR_{10}}$ sequences $u_j^n(m_j|m_{j-1})\sim\prod ^n_{i=1}p(u_{i}|u_{r,i}) $, $m_{j} \in [1:2^{nR_{10}}]$.
\item For each $(m_{j-1},m_j)$, independently generate $2^{nbR_{11}}$ sequences $x_j^n(m|m_j,m_{j-1})\sim\prod ^n_{i=1}p(x_{i}|$ $u_{i},u_{r,i}) $,   $m \in [1:2^{nbR_{11}}]$.
\item For each $m_{j-1}$, independently generate $2^{n\hat{R}}$ sequences $x_{r,j}^n(k_{j-1}|m_{j-1})\sim\prod ^n_{i=1}p(x_{r,i}|u_{r,i}) $,  $k_{j-1} \in [1:2^{n\hat{R}}]$.
\item For each $(m_{j-1},m_j,k_{j-1})$, independently generate $2^{n\hat{R}}$ sequences  $\hat{y}_{r,j}^n(k_j|k_{j-1},m_{j-1},m_j)\sim\prod ^n_{i=1}p(\hat{y}_{r,i}|x_{r,i},u_{r,i},u_{i}) $,   $k_{j} \in [1:2^{n\hat{R}}]$.
\end{itemize}
\subsubsection{Encoding}
In block $j$, the source sends $x_j^n(m|m_j,m_{j-1})$. Assume that the relay has successfully found   compression index $k_{j-1}$ and decoded message $m_{j-1}$ of the previous block, it then sends $x_{r,j}^n(k_{j-1}|m_{j-1})$.

\subsubsection{Decoding at the relay}
At the end of block $j$, upon receiving $y^n_{r,j}$, the relay finds a $\hat{k}_j$ and a unique $\hat{m}_j$ such that
\begin{align}
\label{decoding_rule_relay}
(u_{r,j}^n(m_{j-1}),u_j^n(\hat{m}_j|m_{j-1}),x_{r,j}^n(k_{j-1}|m_{j-1}),&\nonumber\\
\hat{y}_{r,j}^n(\hat{k}_j|k_{j-1},m_{j-1},\hat{m}_j),{y}^n_{r,j}&)\in T^{(n)}_{\epsilon},
\end{align}
where $T^{(n)}_{\epsilon}$ denotes the strong typical set \cite{gamal2010lecture}. By the covering lemma and standard analysis, $P_e \rightarrow 0$ as $n\rightarrow \infty$ if
\begin{align}
\hat{R}& > I(\hat{Y}_r;Y_r|U,U_r,X_r)\nonumber \\
\label{Rhatone}
\hat{R}+R_{10}&\leq I(Y_r;\hat{Y}_r,U|U_r,X_r).
\end{align}



\subsubsection{Decoding at the destination}
At the end of each block $j$, the destination finds the unique $\hat{m}_{j-1}$ such that
\begin{align*}
(u_{j-1}^n(\hat{m}_{j-1}|m_{j-2}),u_{r,j-1}^n(m_{j-2}),&\nonumber\\
x_{r,j-1}^n(k_{j-2}|m_{j-2}),{y}^n_{j-1})&\in T^{(n)}_{\epsilon}\\
\text{and~~~~~~~~~~~}  (u_{r,j}^n(\hat{m}_{j-1}),{y}^n_{j})&\in T^{(n)}_{\epsilon}.
\end{align*}
Following standard  analysis, $P_e \rightarrow 0$ as $n\rightarrow \infty$ if
\begin{align}
\label{R10one}
R_{10} \leq I(U;Y|U_r,X_r)+I(U_r;Y).
\end{align}\par
At the end of block $b$, it finds the unique $\hat{m}$ such that
\begin{align*}
(u_{r,j}^n(m_{j-1}),u_j^n(m_j|m_{j-1}),x_{r,j}^n(\hat{k}_{j-1}|m_{j-1}),&\\
x_j^n(\hat{m}|m_j,m_{j-1}),\hat{y}_{r,j}^n(\hat{k}_j|\hat{k}_{j-1},m_{j-1},m_j),{y}^n_{j}&)\in T^{(n)}_{\epsilon}
\end{align*}
for all $j\in [1:b]$ and some vector $\bold{\hat{k}_j}\in[1:2^{n\hat{R}}]^b$.
As in \cite{sung2011noisy}, $P_e \rightarrow 0$ as $n\rightarrow \infty$ if
\begin{align}
\label{R11one}
R_{11}&\leq \min \{I_1,I_2-\hat{R}\},\quad \text{where}\\
I_1&=I(X;Y,\hat{Y}_r|U,U_r,X_r)\nonumber\\
I_2&=I(X,X_r;Y|U,U_r)+I(\hat{Y}_r;Y,X|X_r,U,U_r).
\end{align}\par
By applying Fourier-Motzkin Elimination to inequalities \eqref{Rhatone}-\eqref{R11one},   the rate in Theorem \ref{rate_one} is achievable.
\end{proof}

\begin{rem}
In relay decoding \eqref{decoding_rule_relay}, we perform joint decoding of both the message and the compression index. If we use sequential decoding to decode the message first and then to find the compression index, we still get the same rate constraints as in Theorem \ref{rate_one}.
\end{rem}

\begin{rem}
We can check that the rate in \eqref{one_way_bounds} is equivalent to the combined DF-CF rate in Theorem 7 \cite{cover1979capacity} for the one-way relay channel. This combined DF-CF scheme is recently shown by Luo et al. \cite{Luo2011on} to outperform both individual schemes in the Gaussian channel for a small range of SNR. However, combined DF-NNC is expected to outperform combined DF-CF for multi-source networks. 
\end{rem}


\begin{rem}
By setting $U_r=X_r,U=X,\hat{Y}_r=0$, the rate in Theorem \ref{rate_one} reduces to the decode-forward relaying rate \cite{cover1979capacity} as
\begin{align}
\label{original_one_way_df}
R \leq \min\{I(X;Y_r|X_r),I(X,X_r;Y)\}
\end{align}
for some $p(x_r)p(x|x_r)p(y,y_r|x,x_r)$.
By setting $U=U_r=0$, it reduces to the NNC rate \cite{sung2011noisy} as
\begin{align*}
R\leq \min \{I(X;Y,\hat{Y}_r|X_r),I(X,X_r;Y)-I(\hat{Y}_r;Y_r|X_r,X,Y)\}
\end{align*}
for some  $p(x)p(x_r)p(y,y_r|x,x_r)p(\hat{y}_r|y_r,x_r).$
\end{rem}

\begin{rem}
The rate constraints in Theorem \ref{rate_one} are similar to those in superposition NNC (Theorem 1 in \cite{Ramalingam2011superposition}), but the code distribution \eqref{one_way_distribution} is a larger set  because of the joint distribution between $(x,u,u_r)$. Hence the achievable rate by the proposed scheme is higher than that in \cite{Ramalingam2011superposition}. Specifically, the scheme in \cite{Ramalingam2011superposition} does not include the decode-forward relaying rate  in  \eqref{original_one_way_df}.
\end{rem}

\subsection{Achievable rate for the Gaussian one-way relay channel}
We now evaluate the achievable rate in Theorem \ref{rate_one} for the Gaussian one-way relay channel as in \eqref{Gaussian_one_model}.
\begin{cor}
The  following rate is achievable for the Gaussian one-way relay channel
\begin{align}
\label{Gaussian_ach_one}
R\leq  \min&\left\{C\left(\frac{g_1^2\beta_1 ^2}{g_1^2\gamma_1 ^2+1}\right),C\left(\frac{(g{\alpha_1}+g_2{\alpha_2})^2+g^2\beta_1 ^2}{g^2\gamma_1 ^2+g_2^2\beta_2 ^2+1}\right)\right\}+\nonumber\\
&~~~C\left(g^2\gamma_1 ^2+\frac{g_1^2\gamma_1 ^2 \cdot g_2^2\beta_2 ^2}{g^2\gamma_1 ^2+g_1^2\gamma_1 ^2+g_2^2\beta_2 ^2+1}\right)
\end{align}
\begin{align}
\label{power_all}
 \text{where} \quad \alpha_1 ^2+\beta_1 ^2+\gamma_1 ^2 \leq P, \quad
\alpha_2 ^2+\beta_2 ^2 \leq P.
\end{align}
\end{cor}

To achieve the rate in \eqref{Gaussian_ach_one}, we set
\begin{align}
U&=\alpha_1 S_1+ \beta_1 S_2,~~~~
X=U+ \gamma_1 S_3\nonumber\\
X_r&= \alpha_2 S_1+ \beta_2 S_4,~~~~
\hat{Y}_r=Y_r+Z'
\end{align}
where $S_1, S_2, S_3, S_4\sim\mathcal{N}(0,1)$ and $Z'\sim\mathcal{N}(0,Q)$ are independent, and the power allocations satisfy constraint \eqref{power_all}.\par

\section{Two-way relay channel}
In this section, we propose a combined scheme based on both decode-forward strategies as in \cite{rankov2006achievable} \cite{xie2007network} and layered noisy network coding \cite{Lim2010layered} for the two-way relay channel. Each user splits its message into three parts: an independent common, a Markov common and a private message. The independent and Markov common messages are encoded differently at the source and are different for each block, both are decoded at both the relay and  destination as in decode-forward. The private message is the same for all blocks and is decoded only at the destination as in noisy network coding.  Each user encodes the Markov common message with block Markov encoding as in \cite{rankov2006achievable}, then superimposes the independent common message on top of it without Markovity, and at last superimposes the private message on top of both.  The relay decodes the two common messages and compresses the rest into two layers: a common   and a refinement layer. In the next block, the relay sends a codeword which encodes the two decoded common messages and two layered compression indices. Then at the end of each block, each user decodes two common messages of the other user by sliding-window decoding over two consecutive blocks. At the end of all blocks, one user uses the information of the common layer to simultaneously decode the private message of the other user, while the other user uses the information of both  the common and refinement layers to decode the other user's private message.\par

\begin{thm}
\label{thm:rate_two}
Let $\mathcal{R}_1$ denote the set of  $(R_1,R_2)$ as follows:
\begin{align}
R_1 &\leq \min\{I_5, I_{12}\}+\min\{I_5-I_1, I_{16}\}\nonumber\\
R_2 &\leq \min\{I_6, I_{14}\}+\min\{I_{17}-I_2, I_{19}\}\nonumber\\
R_1+R_2 &\leq \min\{\min\{I_5,I_{12}\}+I_{15}+I_{18}-I_2+\min\{I_6,I_{14}\},\nonumber\\&~~~~I_{10}+I_{15}+I_{18}-I_2,\nonumber\\
&~~~~I_{10}+\min\{I_{15}-I_{1},I_{16}\}+\min\{I_{17}-I_2,I_{19}\}\}\nonumber\\
2R_1+R_2 &\leq \min\{I_5,I_{12}\}+I_{15}+I_{18}-I_2+I_{10}\nonumber\\
\label{rate_two_1}
&~~~+\min\{I_{15}-I_{1},I_{16}\}
\end{align}
for some joint distribution
\begin{align}
\label{joint_dis}
P^* \triangleq &p(w_1)p(u_1|w_1)p(v_1|w_1,u_1)p(x_1|w_1,u_1,v_1)p(w_2)\nonumber\\
&p(u_2|w_2)p(v_2|w_2,u_2)p(x_2|w_2,u_2,v_2)p(v_r|w_1,w_2)\nonumber\\
&p(u_r|v_r,w_1,w_2)
p(x_r|u_r,v_r,w_1,w_2)\nonumber\\
&p(\hat{y}_r,\tilde{y}_r|y_r,x_r,u_r,v_r,w_1,w_2,u_1,v_1,u_2,v_2),
\end{align}
where $I_{j}$ are defined in \eqref{rate_two_first}-\eqref{rate_two_last}, then
$\mathcal{R}_1$ is achievable if user 2 only uses the common layer, while user 1 uses both the common   and refinement layers. If the two users exchange   decoding layers, they can achieve a corresponding set  $\mathcal{R}_2$. By time sharing, the convex hull of  $\mathcal{R}_1 \cup \mathcal{R}_2$ is achievable.
\end{thm}

\begin{proof}
We use a block coding scheme in which each user sends $b-1$ messages over $b$ blocks of $n$ symbols each.
\subsubsection{Codebook generation}
Fix a joint distribution $P^*$ as in \eqref{joint_dis}. Each user $l \in \{1,2\}$ splits its message into three parts: $m_{l0}, m_{l1}$ and $m_{l2}$. For each $j\in[1:b]$ and $l \in \{1,2\}$
\begin{itemize}
\item Independently generate $2^{nR_{l0}}$ sequences $w_{l,j}^n(m_{l0,j-1})\sim\prod ^n_{i=1}p(w_{l,i})$,  $m_{l0,j-1} \in [1:2^{nR_{l0}}]$.
\item For each $m_{l0,j-1}$, independently generate $2^{nR_{l0}}$ sequences $u_{l,j}^n(m_{l0,j}|m_{l0,j-1})\sim\prod ^n_{i=1}p(u_{l, i}|w_{l, i})$, $m_{l0,j} \in [1:2^{nR_{l0}}]$.
\item For each $m_{l0,j-1},m_{l0,j}$, independently generate $2^{nR_{l1}}$ sequences $v_{l,j}^n(m_{l1,j}|m_{l0,j},m_{l0,j-1})\sim\prod ^n_{i=1}p(v_{l, i}|u_{l, i},w_{l, i})$,   $m_{l1,j} \in [1:2^{nR_{l1}}]$.
\item For each $m_{l0,j-1},m_{l0,j},m_{l1,j}$, independently generate $2^{nbR_{l2}}$ sequences $x_{l,j}^n(m_{l2}|m_{l1,j},m_{l0,j},m_{l0,j-1})\sim\prod ^n_{i=1}p(x_{l, i}|v_{l, i},u_{l, i},w_{l, i})$,   $m_{l2} \in [1:2^{nbR_{l2}}]$.
\item For each pair $(m_{10,j-1},m_{20,j-1})$, independently generate $2^{n(R_{11}+R_{21})}$ sequences $v_r^{n}(K|m_{10,j-1},m_{20,j-1})\sim\prod ^n_{i=1}p(v_{ri}|w_{1, i},w_{2, i}) $, where $K \in [1:2^{n(R_{11}+R_{21})}]$. Map each pair  $(m_{11,j-1},m_{21,j-1})$ to one $K$.
\item For each vector $\bold{m}_{j-1}=(m_{10,j-1},m_{20,j-1},$ $m_{11,j-1},m_{21,j-1})$, independently generate $2^{n\tilde{R}}$ sequences $u_{r,j}^n(t_{j-1}|\bold{m}_{j-1})\sim\prod ^n_{i=1}p(u_{r, i}|v_{r, i},w_{1, i},w_{2, i})$,   $t_{j-1} \in [1:2^{n\tilde{R}}]$.
\item For each $(t_{j-1},\bold{m}_{j-1})$, independently generate $2^{n\hat{R}}$ sequences $x_{r,j}^n(l_{j-1}|t_{j-1},\bold{m}_{j-1})\sim \prod ^n_{i=1}p(x_{r, i}|u_{r, i},v_{r, i},$ $w_{1, i},w_{2, i})$,   $l_{j-1} \in [1:2^{n\hat{R}}]$.
\item For each $(t_{j-1},\bold{m}_{j-1},\bold{m}_{j})$, independently generate $2^{n\tilde{R}}$ sequences $\tilde{y}_{r,j}^n(t_{j}|t_{j-1},\bold{m}_{j-1},\bold{m}_{j})\sim \prod ^n_{i=1}p(\tilde{y}_{r, i}|$ $u_{r, i},v_{r, i},w_{1, i},w_{2, i},u_{1, i},u_{2, i},v_{1, i},v_{2, i})$,   $t_{j} \in [1:2^{n\tilde{R}}]$.
\item For each $(t_{j},t_{j-1},l_{j-1},\bold{m}_{j-1},\bold{m}_{j})$, independently generate $2^{n\hat{R}}$ sequences $\hat{y}_{r,j}^n(l_{j}|l_{j-1},t_{j},t_{j-1},\bold{m}_{j-1},\bold{m}_{j})\sim \prod ^n_{i=1}p(\hat{y}_{r, i}|\tilde{y}_{r, i},x_{r, i},u_{r, i},v_{r, i},w_{1, i},w_{2, i},u_{1, i},u_{2, i},v_{1,i},$ $v_{2, i})$,   $t_{j} \in [1:2^{n\tilde{R}}]$.
\end{itemize}

\subsubsection{Encoding}
In block $j$, user $l\in\{1,2\}$ sends $x_{l,j}^n(m_{l2}|m_{l1,j},m_{l0,j},m_{l0,j-1})$.  
Let $\bold{m}_j=(m_{10,j},m_{20,j},$ $m_{11,j},m_{21,j})$. At the end of block $j$, the relay has decoded $\bold{m}_{j-1},\bold{m}_j$.
Upon receiving $y_{r,j}^n$, it finds an index pair $(\hat{t}_j,\hat{l}_j)$ such that
\begin{align*}
(\hat{y}_{r,j}^n(\hat{l}_{j}|l_{j-1},\hat{t}_{j},t_{j-1},\bold{m}_{j-1},\bold{m}_j),\tilde{y}_{r,j}^n(\hat{t}_{j}|t_{j-1},\bold{m}_{j-1},\bold{m}_j),&\\x_{r,j}^n(l_{j-1}|t_{j-1},\bold{m}_{j-1}),u_{r,j}^n(t_{j-1}|\bold{m}_{j-1}),
w_{1,j}^n,w_{2,j}^n,&\\v_{r,j}^n,u_{1,j}^n,v_{1,j}^n,u_{2,j}^n,v_{2,j}^n,y_{r,j}^n)\in T^{(n)}_{\epsilon}&.
\end{align*}
According to Lemma 1 in \cite{Lim2010layered}, the probability that no such $(\hat{t}_j,\hat{l}_j)$ exists goes to $0$  as $n\rightarrow \infty $ if
\begin{align}
\label{rate_two_first}
&\tilde{R}>I(\tilde{Y}_r;X_r,Y_r|U_r,V_r,U_1,V_1,U_2,V_2,W_1,W_2)\triangleq I_1 \\
&\tilde{R}+\hat{R}>I(\tilde{Y}_r;X_r,Y_r|U_r,V_r,U_1,V_1,U_2,V_2,W_1,W_2)+\nonumber\\
&~~~~~~~I(\hat{Y}_r;Y_r|\tilde{Y}_r,X_r,U_r,V_r,U_1,V_1,U_2,V_2,W_1,W_2)\triangleq I_2.\nonumber
\end{align}
The relay then sends
$x_{r,j+1}^n(l_{j}|t_{j},\bold{m}_j)$ at block $j+1$.

\subsubsection{Relay decoding}
At the end of block $j$, the relay finds the unique $(\hat{m}_{10,j},\hat{m}_{20,j},\hat{m}_{11,j},\hat{m}_{21,j})$ such that
\begin{align*}
\!\!\!\!\!\!\!\!(w_{1,j}^n(m_{10,j-1}),u_{1,j}^n(\hat{m}_{10,j}|m_{10,j-1}),v_{1,j}^n(\hat{m}_{11,j}|\hat{m}_{10,j},m_{10,j-1}),&\\
\!\!\!\!\!\!\!\!w_{2,j}^n(m_{20,j-1}),u_{2,j}^n(\hat{m}_{20,j}|m_{20,j-1}),v_{2,j}^n(\hat{m}_{21,j}|\hat{m}_{20,j},m_{20,j-1}),&\\
v_{r,j}^n(m_{11,j-1},m_{21,j-1}|m_{10,j-1},m_{20,j-1}),y^n_{r,j}) \in T^{(n)}_{\epsilon}.&
\end{align*}
As in the multiple access channel, $P_e \rightarrow 0$ as $n\rightarrow \infty$ if
\begin{align}
R_{11} &\leq I(V_1;Y_r|V_r,W_1,U_1,W_2,U_2,V_2)\triangleq I_3\nonumber\\
R_{21} &\leq I(V_2;Y_r|V_r,W_2,U_2,W_1,U_1,V_1)\triangleq I_4\nonumber\\
R_{10}+R_{11} &\leq I(U_1,V_1;Y_r|V_r,W_1,W_2,U_2,V_2)\triangleq I_5\nonumber\\
R_{20}+R_{21} &\leq I(U_2,V_2;Y_r|V_r,W_2,W_1,U_1,V_1)\triangleq I_6\nonumber\\
R_{11}+R_{21} &\leq I(V_1,V_2;Y_r|V_r,W_1,U_1,W_2,U_2)\triangleq I_7\nonumber\\
R_{10}+R_{11}+R_{21} &\leq I(U_1,V_1,V_2;Y_r|V_r,W_1,W_2,U_2)\triangleq I_8\nonumber\\
R_{20}+R_{11}+R_{21} &\leq I(U_2,V_1,V_2;Y_r|V_r,W_1,W_2,U_1)\triangleq I_9\nonumber\\
R_{10}+R_{20}+R_{11}&+R_{21} \leq\nonumber\\
&\!\!\!\!\!\!\!(U_1,V_1,U_2,V_2;Y_r|V_r,W_1,W_2)\triangleq I_{10}.
\end{align}

\subsubsection{User decoding}
At the end of block $j$, user 2 finds the unique $(\hat{m}_{10,j-1},\hat{m}_{11,j-1})$ such that
\begin{align*}
(u_{1,j-1}^n(\hat{m}_{10,j-1}|m_{10,j-2}),v_{1,j-1}^n(\hat{m}_{11,j-1}| \hat{m}_{10,j-1},m_{10,j-2}),&\\
w_{1,j-1}^n,w_{2,j-1}^n,u_{2,j-1}^n,v_{2,j-1}^n,x_{2,j-1}^n,
v_{r,j-1}^n,y^n_{2,j-1}) \in T^{(n)}_{\epsilon}&\\
\text{and~~~~~~~~~} (w_{1,j}^n(\hat{m}_{10,j-1}),
w_{2,j}^n,u_{2,j}^n,v_{2,j}^n,x_{2,j}^n,&\\
v_{r,j}^n(\hat{m}_{11,j-1},m_{21,j-1}|\hat{m}_{10,j-1},m_{20,j-1}),y^n_{2,j}) \in T^{(n)}_{\epsilon}.&
\end{align*}
The error probability  goes to $0$ as $n\rightarrow \infty $ if
\begin{align}
R_{11}&\leq I(V_1;Y_2|V_r,W_1,U_1,W_2,U_2,V_2,X_2)\\
&~~~+I(V_r;Y_2|W_1,W_2,U_2,V_2,X_2)\triangleq I_{11}\nonumber\\
R_{10}+R_{11}&\leq I(U_1,V_1;Y_2|V_r,W_1,W_2,U_2,V_2,X_2)\nonumber\\
&~~~+I(W_1,V_r;Y_2|W_2,U_2,V_2,X_2)\nonumber\\
&=I(W_1,U_1,V_1,V_r;Y_2|W_2,U_2,V_2,X_2)\triangleq I_{12}.\nonumber
\end{align}
Similarly, for vanishing-error  decoding at user 1
\begin{align}
R_{21}&\leq I(V_2;Y_1|V_r,W_2,U_2,W_1,U_1,V_1,X_1)\\
&~~~+I(V_r;Y_1|W_2,W_1,U_1,V_1,X_1)\triangleq I_{13}\nonumber\\
R_{20}+R_{21}&\leq I(W_2,U_2,V_2,V_r;Y_1|W_1,U_1,V_1,X_1)\triangleq I_{14}.\nonumber
\end{align}\par
At the end of last block $b$, user 2 uses one compression layer to find the unique $\hat{m}_{12}$ such that
\begin{align*}
(x_{1,j}^n(\hat{m}_{12}|m_{11,j},m_{10,j},m_{10,j-1}),v_{1,j}^n,u_{1,j}^n,w_{1,j}^n,
x_{2,j}^n,v_{2,j}^n,u_{2,j}^n,&\\
w_{2,j}^n,
v_{r,j}^n,u_{r,j}^n(\hat{t}_{j-1}|\bold{m}_{j-1}),
\tilde{y}_{r,j}^n(\hat{t}_{j}|\hat{t}_{j-1},\bold{m}_{j-1},\bold{m}_{j}),y^n_{2,j}) \in T^{(n)}_{\epsilon}&
\end{align*}
for all $j\in[1:b]$ and some vector $\bold{\hat{t}_{j}}\in[1:2^{n\tilde{R}}]^b$.
As in \cite{Lim2010layered}, $P_e \rightarrow 0$ as $n\rightarrow \infty$ if
\begin{align}
&R_{12}+\tilde{R}\leq I(X_1,U_r;Y_2|X_2,W_1,U_1,V_1,W_2,U_2,V_2,V_r)+\nonumber\\
& I(\tilde{Y}_r;X_1,X_2,Y_2|U_r,W_1,U_1,V_1,W_2,U_2,V_2,V_r)\triangleq I_{15} \\
&R_{12}\leq I(X_1;\tilde{Y}_r,Y_2|X_2,U_r,W_1,U_1,V_1,W_2,U_2,V_2,V_r)\triangleq I_{16}\nonumber
\end{align}\par
Using both layers, user 1 finds the unique $\hat{m}_{22}$ such that
\begin{align*}
&~~(x_{1,j}^n,v_{1,j}^n,u_{1,j}^n,w_{1,j}^n,x_{2,j}^n(\hat{m}_{22}|m_{21,j},m_{20,j},m_{20,j-1}),y^n_{1,j}\\
&~~v_{2,j}^n,u_{2,j}^n,w_{2,j}^n
v_{r,j}^n,u_{r,j}^n(\hat{t}_{j-1}|\bold{m}_{j-1}),x_{r,j}^n(\hat{l}_{j-1}|\hat{t}_{j-1},\bold{m}_{j-1}),\\ &~~\tilde{y}_{r,j}^n(\hat{t}_{j}|\hat{t}_{j-1},\bold{m}_{j-1},\bold{m}_{j}),\hat{y}_{r,j}^n(\hat{l}_{j}|\hat{t}_{j},\hat{l}_{j-1},\hat{t}_{j-1},\bold{m}_{j-1},\bold{m}_{j})) \in T^{(n)}_{\epsilon}
\end{align*}
for all $j\in[1:b]$ and   some vectors $\bold{\hat{t}_{j}}\in[1:2^{n\tilde{R}}]^b, \bold{\hat{l}_{j}}\in[1:2^{n\hat{R}}]^b$.
As in \cite{Lim2010layered}, $P_e \rightarrow 0$ as $n\rightarrow \infty$ if
\begin{align}
\label{rate_two_last}
&R_{22}+\tilde{R}+\hat{R}\leq I(X_2,X_r;Y_1|X_1,W_1,U_1,V_1,W_2,U_2,V_2,V_r)\nonumber\\
&+I(\hat{Y}_r;X_1,X_2,Y_1|\tilde{Y}_r,X_r,U_r,W_1,U_1,V_1,W_2,U_2,V_2,V_r)\nonumber\\
&+I(\tilde{Y}_r;X_1,X_2,X_r,Y_2|U_r,W_1,U_1,V_1,W_2,U_2,V_2,V_r)\triangleq I_{17}\nonumber\\
&\!\!\!\!R_{22}+\hat{R}\leq I(X_2,X_r;Y_1,\tilde{Y}_r|X_1,U_r,W_1,U_1,V_1,W_2,U_2,V_2,V_r)\nonumber\\
&\!\!\!\!+I(\hat{Y}_r;X_1,X_2,Y_1|\tilde{Y}_r,X_r,U_r,W_1,U_1,V_1,W_2,U_2,V_2,V_r)\triangleq I_{18}\nonumber\\
&R_{22}\leq I(X_2;\tilde{Y}_r,\hat{Y}_r,Y_1|X_1,U_r,X_r,W_1,U_1,V_1,W_2,U_2,V_2,V_r)\nonumber\\ &~~~~~\triangleq I_{19}.
\end{align}\par
By applying Fourier-Motzkin Elimination to inequalities \eqref{rate_two_first}-\eqref{rate_two_last}, the rate region in Theorem \ref{thm:rate_two} is achievable.
\end{proof}

\begin{rem}
By   setting $V_1=W_2=U_2=V_2=X_2=U_r=V_r=\tilde{Y}_r=0$, we obtain the rate for the one-way channel in Theorem \ref{rate_one} from the region in Theorem \ref{thm:rate_two}.
\end{rem}

\begin{rem}
The proposed combined DF-LNNC scheme includes each schemes in \cite{rankov2006achievable, xie2007network, Lim2010layered} as a special case. Specifically, it reduces to the scheme in \cite{rankov2006achievable} by setting $U_1=X_1,U_2=X_2,V_r=X_r,V_1=V_2=U_r=\hat{Y}_r=\tilde{Y}_r=0$, to the scheme in \cite{xie2007network} by setting $V_1=X_1,V_2=X_2,V_r=X_r,W_1=U_1=W_2=U_2=U_r=\hat{Y}_r=\tilde{Y}_r=0$, and to the scheme in \cite{Lim2010layered} by setting $W_1=U_1=V_1=W_2=U_2=V_2=V_r=0$.
\end{rem}

\begin{rem}
In our proposed scheme, the Markov common messages bring a coherent gain between the source and the relay, but they require the relay to split its power for each message because of superposition coding. For the independent common messages, the relay can use its whole power to send their bin index, which can then solely represent one message when decoding because of side information on the other message at each destination.
\end{rem}

\begin{rem}
Rate region for the Gaussian TWRC can be obtained by applying Theorem \ref{thm:rate_two} with the following signaling:
\begin{align}
X_1&=\alpha_1 S_1+\beta_1 S_2 +\gamma_1 S_3+\delta_1 S_4\nonumber\\
X_2&=\alpha_2 S_5+\beta_2 S_6 +\gamma_2 S_7+\delta_2 S_8\nonumber\\
X_r&=\alpha_{31} S_1 +\alpha_{32} S_5+ \gamma_3 S_9+\beta_3 S_{10}+\delta_3 S_{11}\nonumber\\
\hat{Y}_r&=Y_r+\hat{Z}_r; \quad
\tilde{Y}_r=\hat{Y}_r+\tilde{Z}_r,
\end{align}
where the power allocations satisfy
\begin{align}
\alpha_1^2+\beta_1^2+\gamma_1^2+\delta_1^2 &\leq P, ~~~~
\alpha_2^2+\beta_2^2+\gamma_2^2+\delta_2^2 \leq P,\nonumber\\
\alpha_{31}^2+\alpha_{32}^2+\beta_3^2+\gamma_3^2+\delta_3^2 &\leq P,
\end{align}
all $S_i\sim\mathcal{N}(0,1)$ and $\hat{Z}_r\sim\mathcal{N}(0,\hat{Q}),\tilde{Z}_r\sim\mathcal{N}(0,\tilde{Q})$ are independent.
The specific rate constraints for the Gaussian channel, however, are omitted because of the lack of space.
\end{rem}

\begin{figure}[t]
    \begin{center}
    \includegraphics[width=0.35\textwidth]{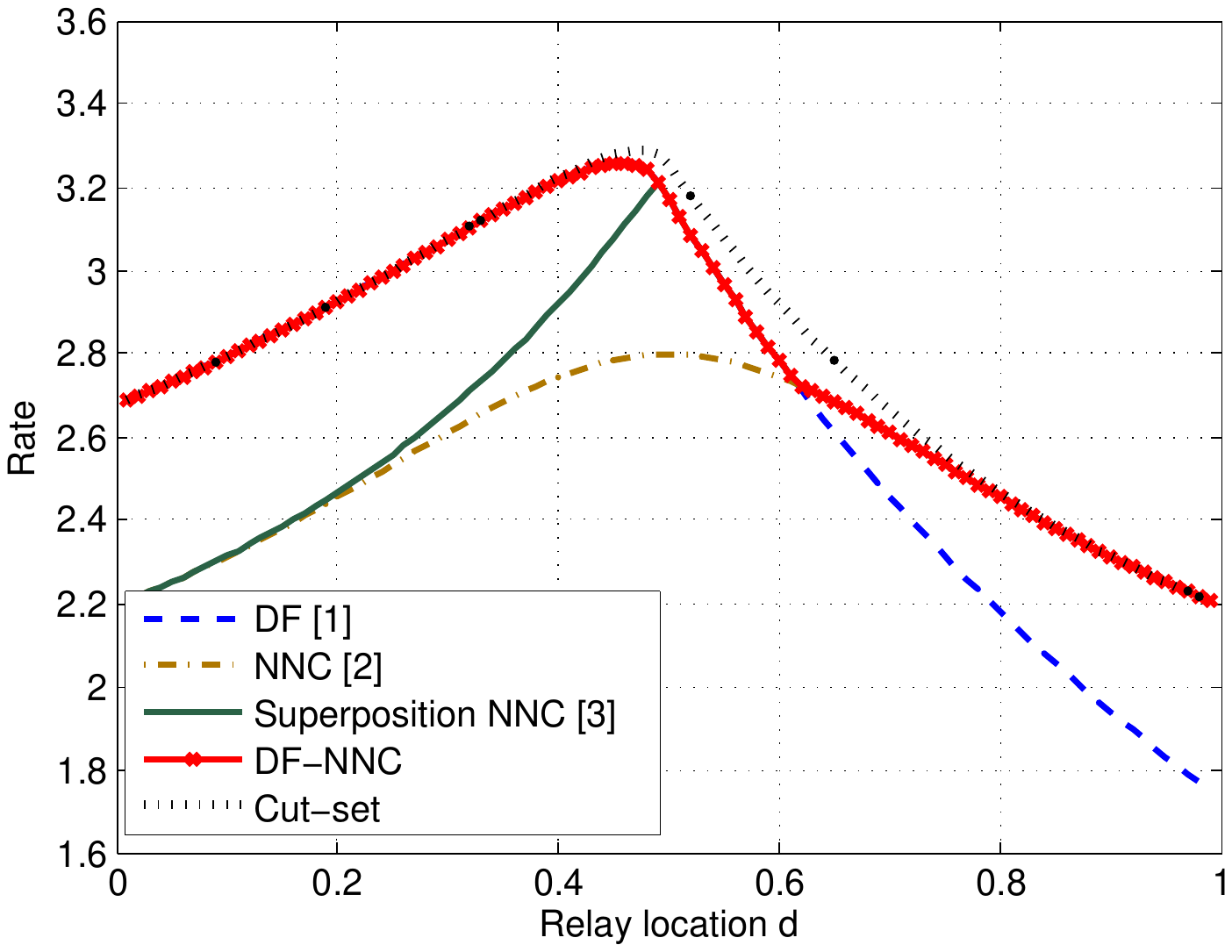}
    \caption{Achievable rate comparison for the one-way relay channel with $P=10, g_1=d^{-\gamma/2},g_2=(1-d)^{-\gamma/2}, g=1, \gamma=3$.} \label{fig:sim1}
    \end{center}
\end{figure}

\section{Numerical Results}
We numerically compare the performance of the proposed combined   schemes with the original DF   and NNC. The rate regions are obtained by exhaustive simulation, but optimization can also be used to achieve the maximum rates.
Consider the Gaussian channels as in \eqref{Gaussian_one_model} and \eqref{GTWRC}. Assume all the nodes are on a straight line. The relay is at a distance $d$ from the source and distance $1-d$ from the destination which makes $g_1=d^{-\gamma/2}$ and $g_2=(1-d)^{-\gamma/2}$, where $\gamma$ is the path loss exponent. Figure \ref{fig:sim1} shows the achievable rate for the one-way relay channel with $P=10, \gamma=3$. The combined DF-NNC scheme supersedes both the DF and NNC schemes. It can achieve the capacity of the one-way relay channel when the relay is close to either the source or the destination.  Figure \ref{fig:sim3} shows the sum rate for the two-way relay channel with $P=10, \gamma=3$. Our proposed scheme achieves larger sum rate than all 3 individual schemes when the relay is close to either user, while reducing to layered NNC when the relay is close to the middle of the two users.
Figure \ref{fig:sim2} shows the achievable rate regions for the Gaussian TWRC using these 4 schemes. The achievable  region of our proposed scheme encompasses all 3 individual schemes.

\begin{figure}[t]
    \begin{center}
    \includegraphics[width=0.35\textwidth]{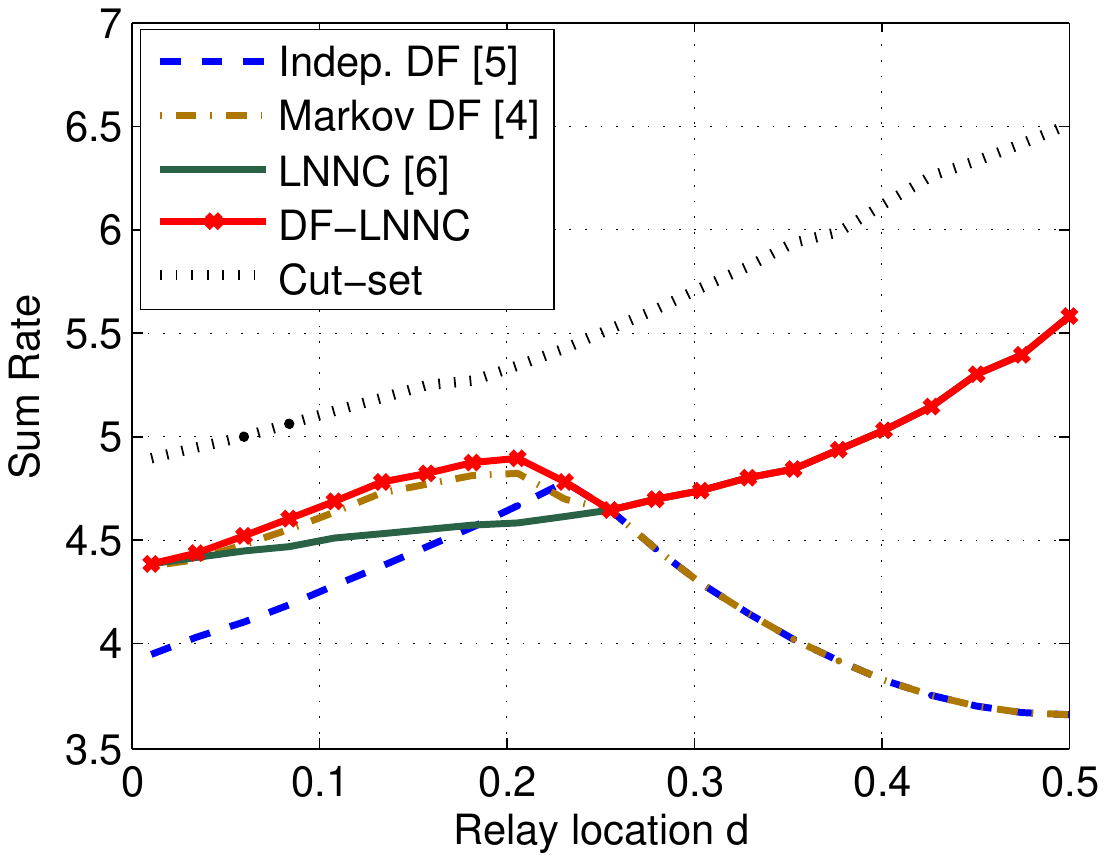}
    \caption{Sum rate for the two-way relay channel with $P=10, g_{r1}=g_{1r}=d^{-\gamma/2},g_{r2}=g_{2r}=(1-d)^{-\gamma/2}, g_{12}=g_{21}=1, \gamma=3$.} \label{fig:sim3}
    \end{center}
\end{figure}

\begin{figure}[t]
    \begin{center}
    \includegraphics[width=0.35\textwidth]{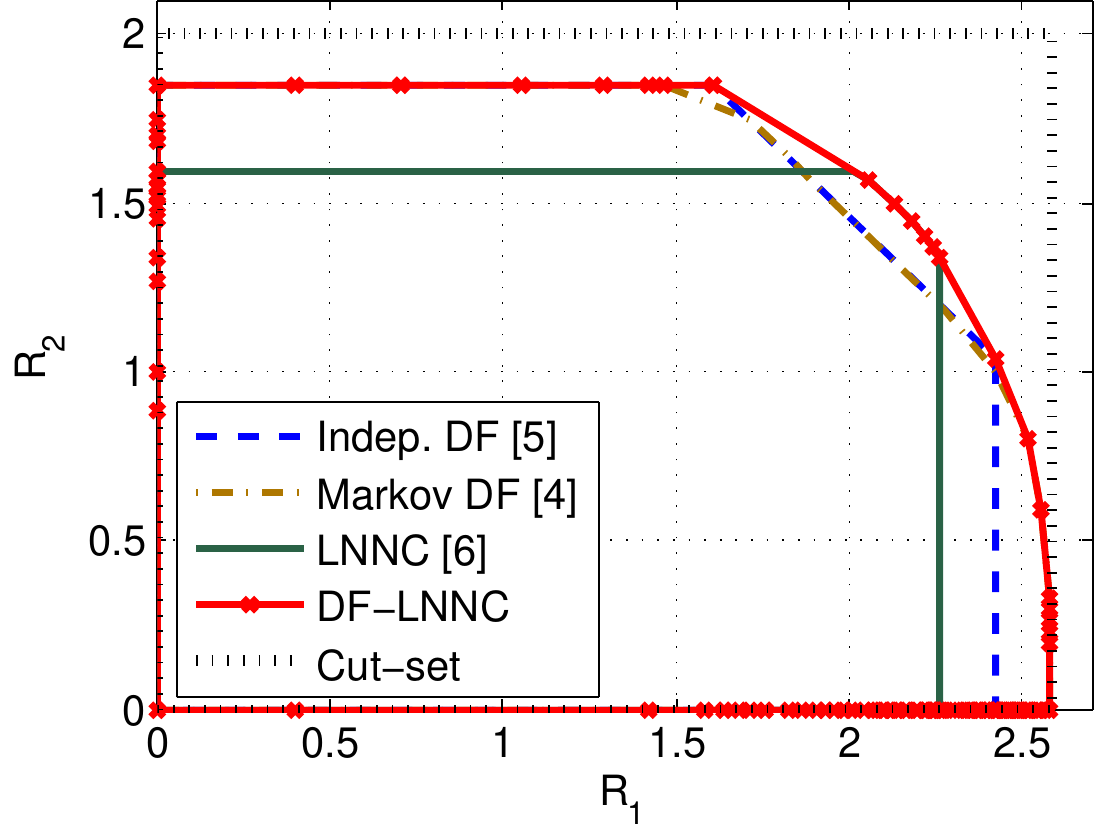}
    \caption{Achievable rate region comparison for the two-way relay channel with $P=3, g_{r1}=6,g_{1r}=2,g_{r2}=2,g_{2r}=3,g_{12}=1,g_{21}=0.5$.} \label{fig:sim2}
    \end{center}
\end{figure}

\section{Conclusion}
We have proposed two combined schemes: DF-NNC for the one-way and  DF-LNNC  for the two-way relay channels. Both  schemes perform message splitting, block Markov encoding, superposition  coding and noisy network coding. Each combined scheme encompasses all respective individual schemes (DF and NNC or LNNC) and strictly outperforms superposition NNC in \cite{Ramalingam2011superposition}.   These are initial results for combining decode-forward and noisy network coding for a multi-source network.

\bibliographystyle{IEEEtran}
\bibliography{reflist}

\end{document}